\documentclass[11pt,letterpaper]{article}
\usepackage{color,ifpdf,latexsym,graphicx,url}
\usepackage[margin=1in]{geometry}
\usepackage{palatino,mathpazo}
\usepackage{amsmath, amssymb, amsthm, amsfonts,alltt,clrscode}
\usepackage{array,colortbl,hhline,xcolor}
\usepackage{subcaption}

\title{Complexity of Motion Planning of Arbitrarily Many Robots: \\
  Gadgets, Petri Nets, and Counter Machines}
\author{%
  Joshua Ani%
    \thanks{MIT Computer Science and Artificial Intelligence Laboratory,
      32 Vassar St., Cambridge, MA 02139, USA,
      \protect\url{{joshuaa,mcoulomb,edemaine,diomidov,tagomez7,dylanhen,jaysonl}@mit.edu}}
\and
  Michael Coulombe\footnotemark[1]
\and
  Erik D. Demaine\footnotemark[1]
\and
  Yevhenii Diomidov\footnotemark[1]
\and
  Timothy Gomez\footnotemark[1]
\and
  Dylan Hendrickson\footnotemark[1]
\and
  Jayson Lynch\footnotemark[1]
}
\date{}

\newif\ifabstract
\abstracttrue
\newif\iffull
\ifabstract \fullfalse \else \fulltrue \fi

\usepackage{hyperref}
\hypersetup{breaklinks,bookmarksnumbered,bookmarksopen,bookmarksopenlevel=2}
{\makeatletter \hypersetup{pdftitle={\@title}}}

{\makeatletter
 \gdef\xxxmark{%
   \expandafter\ifx\csname @mpargs\endcsname\relax 
     \expandafter\ifx\csname @captype\endcsname\relax 
       \marginpar{xxx}
     \else
       xxx 
     \fi
   \else
     xxx 
   \fi}
 \gdef\xxx{\@ifnextchar[\xxx@lab\xxx@nolab}
 \long\gdef\xxx@lab[#1]#2{\textbf{[\xxxmark #2 ---{\sc #1}]}}
 \long\gdef\xxx@nolab#1{\textbf{[\xxxmark #1]}}
}

{\makeatletter \gdef\fps@figure{!htbp}}

\let\realbfseries=\bfseries
\def\bfseries{\realbfseries\boldmath}

\newtheorem{theorem}{Theorem}[section]
\newtheorem{lemma}[theorem]{Lemma}

\newtheorem{definition}[theorem]{Definition}

\let\epsilon=\varepsilon
\def\defn#1{\textbf{\textit{\boldmath #1}}}

\graphicspath{{figures/}}

\begin{document}

\maketitle

\begin{abstract}
  We extend the motion-planning-through-gadgets framework to several new
  scenarios involving various numbers of robots/agents,
  and analyze the complexity of the resulting motion-planning problems.
  While past work considers just one robot or one robot per player,
  most of our models allow for one or more locations to \emph{spawn} new robots
  in each time step, leading to arbitrarily many robots.
  In the 0-player context, where all motion is deterministically forced,
  we prove that deciding whether any robot ever reaches a specified location
  is undecidable, by representing a counter machine.
  In the 1-player context, where the player can choose how to move the robots,
  we prove equivalence to Petri nets, EXPSPACE-completeness for
  reaching a specified location, PSPACE-completeness for reconfiguration, and
  ACKERMANN-completeness for reconfiguration when robots can be destroyed in
  addition to spawned.
  Finally, we consider a variation on the standard 2-player context where,
  instead of one robot per player, we have one robot shared by the players,
  along with a ko rule to prevent immediately undoing the previous move.
  We prove this impartial 2-player game EXPTIME-complete.
\end{abstract}


\section{Introduction}

Intuitively, motion planning is harder with more agents/robots.
This paper formalizes this intuition by studying the effects of
varying the number of robots in a recent combinatorial model for
combinatorial motion planning and the resulting computational complexity.

Specifically, the \defn{motion-planning-through-gadgets framework}
was introduced in 2018 \cite{gadgets} and has had significant study since
\cite{gadgets2,doors,ani2022trains,ani2022traversability,demaine2022pspace,ani2022checked,lynch2020framework,hendrickson2021gadgets}.
In the original one-player setting, the framework considers a single agent/robot
traversing a dynamic network of ``gadgets'', where each gadget has finite state
and a finite set of traversals that the robot can make depending on the state,
and each traversal potentially changes the state
(and thus which future traversals are possible).
The goal is for the robot to traverse from one specified location to another
(\defn{reachability}), or for the system of gadgets to reach a desired state
(\defn{reconfiguration}) \cite{ani2022traversability}.
Existing results characterize in many settings which gadgets
(in many cases, one extremely simple gadget)
result in NP-complete or PSPACE-complete motion-planning problems,
and which gadgets are simple enough to admit polynomial-time motion planning.
This framework has already proved useful for analyzing the computational
complexity of motion-planning problems involving modular robots
\cite{a2021characterizing}, swarm robots
\cite{balanza2019full,caballero2020relocating},
and chemical reaction networks \cite{alaniz2022reachability}.
These applications all involve naturally multi-agent systems,
so it is natural to consider how the complexity of the gadgets framework
changes with more than one robot.

\paragraph{1-player with arbitrarily many robots.}
In Section~\ref{sec:1playerInf},
we consider a generalization of this 1-player gadget model
to an arbitrary number of robots,
and the player can move any one robot at a time.
By itself, this extension does not lead to additional computational complexity:
such motion planning remains in PSPACE, or in NP if each gadget can be
traversed a limited number of times.
To see the true effect of an arbitrary number of robots, we add one or two
additional features: a \defn{spawner} gadget that can create new robots,
and optionally a \defn{destroyer} gadget that can remove robots.
For reachability, only the spawning ability matters --- it is equivalent to
having one ``source'' location with infinitely many robots ---
and we show that the complexity of motion planning grows to
EXPTIME-complete with a simple single gadget
called the \defn{symmetric self-closing door}
(previously shown PSPACE-complete without spawners~\cite{doors}).
For reconfiguration, we show that motion planning with a spawner and
symmetric self-closing door is just PSPACE-complete
(just like without a spawner), but when we add a destroyer,
the complexity jumps to ACKERMANN-complete
(in particular, the running time is not elementary).
These results follow from a general equivalence to \defn{Petri nets} ---
a much older and well-studied model of dynamic systems ---
whose complexity has very recently been characterized
\cite{leroux2022reachability, czerwinski2022reachability}. 

\paragraph{0-player with arbitrarily many robots.}
In Section~\ref{sec:spawners}, we consider the same concepts in a
0-player setting, where every robot has a forced traversal during its turn,
and spawners and robots take turns in a round-robin schedule.
0-player motion planning in the gadget framework with one robot
was considered previously \cite{ani2022trains,demaine2022pspace},
with the complexity naturally maxing out at PSPACE-completeness.
With spawners and a handful of simple gadgets, we prove that the
computational complexity of motion planning
increases all the way to RE-completeness.
In particular, the reachability problem becomes undecidable.
This is a surprising contrast to the 1-player setting described above,
where the problem is decidable.

\paragraph{Impartial 2-player with a shared robot.}
In Section~\ref{sec:impartial}, we consider changing the number of robots
in the downward direction.
Past study of 2-player motion planning in the gadget framework \cite{gadgets2}
considers one robot per player, with each player controlling their own robot.
What happens if there is instead only one robot, shared by the two players?
This variant results in an \defn{impartial} game where the possible moves
in a given state are the same no matter which player moves next.
To prevent one player from always undoing the other player's moves,
we introduce a \defn{ko rule}, which makes it illegal to perform two
consecutive transitions in the same gadget.
In this model, we show that 2-player motion planning is
EXPTIME-complete for a broad family of gadgets called
``reversible deterministic interacting $k$-tunnel gadget'',
matching a previous result for 2-player motion planning
with one robot per player \cite{gadgets2}.
In other words, reducing the number of robots in this way
does not affect the complexity of the problem
(at least for the gadgets understood so far).

\section{Standard Gadget Model}

We now define the gadget model of motion planning, introduced in \cite{gadgets}.

In general, a \defn{gadget} consists of a finite number of
\defn{locations} (entrances/exits) and a finite number of \defn{states}.
Each state $S$ of the gadget defines a labeled directed graph on the locations,
where a directed edge $(a,b)$ with label $S^\prime$ means that a robot
can enter the gadget at location $a$ and exit at location $b$, changing the state of the gadget from $S$ to $S^\prime$.
Equivalently, a gadget is specified by its \defn{transition graph},
a directed graph whose vertices are state/location pairs,
where a directed edge from $(S,a)$ to $(S',b)$ represents that the robot
can traverse the gadget from $a$ to $b$ if it is in state~$S$,
and that such traversal will change the gadget's state to~$S^\prime$.
Gadgets are \defn{local} in the sense that traversing a gadget does
not change the state of any other gadgets.

A \defn{system of gadgets} consists of gadgets, their initial states, and
a \defn{connection graph} on the gadgets' locations.
If two locations $a$ and $b$ of two gadgets (possibly the same gadget) are connected
by a path in the connection graph, then a robot can traverse freely between
$a$ and~$b$ (outside the gadgets).
(Equivalently, we can think of locations $a$ and $b$ as being identified,
effectively contracting connected components of the connection graph.)
These are all the ways that the robot can move: exterior to gadgets using
the connection graph, and traversing gadgets according to their current states.

Previous work has focused on the robot reachability\footnote{In \cite{gadgets,gadgets2}, ``reachability'' refers to whether an agent/robot can reach a target location. Here we refer to it as \emph{robot reachability} since for models such as Petri-nets the Reachability problem refers to whether a full configuration is reachable.} problem \cite{gadgets,gadgets2}:

\begin{definition}
	For a gadget $G$, \defn{robot reachability for $G$} is the following decision problem. Given a system of gadgets consisting of copies of $G$, the starting location(s), and a win location, is there a path a robot can take from the starting location to the win location?
\end{definition}

Gadget reconfiguration, which had target states for the gadgets to be in, was considered in \cite{ani2022traversability} and \cite{hendrickson2021gadgets}. We additionally investigate a problem where we have target states and multiple locations which require specific numbers of robots.
	
\begin{definition}
	For a gadget $G$, the \defn{multi-robot targeted reconfiguration problem for $G$} is the following decision problem. Given a system of gadgets consisting of copies of $G$, the starting location(s), and
	a target configuration of gadgets and robots, 
	is there a sequence of moves the robots can take to reach the target configuration? 
\end{definition}

\cite{gadgets2} also defines 2-player and team analogues of this problem. In this case, each player has their own starting and win locations, and the players take turns making a single transition across a gadget (and any movement in the connection graph). The winner is the player who reaches their win location first. The decision problem is whether a particular player or team can force a win. When there are multiple robots, we are asking whether any of them can reach the win location.

We will consider several specific classes of gadgets.

\begin{definition}
	A \defn{$k$-tunnel} gadget has $2k$ locations, which are partitioned into $k$ pairs called \defn{tunnels}, such that every transition is between two locations in the same tunnel.
\end{definition}

Most of the gadgets we consider are $k$-tunnel.

\begin{definition}
	The \defn{state-transition graph} of a gadget is the directed graph which has a vertex for each state, and an edge $S\to S^\prime$ for each transition from state $S$ to $S^\prime$. A \defn{DAG} gadget is a gadget whose state-transition graph is acyclic.
\end{definition}

DAG gadgets naturally lead to bounded problems, since they can be traversed a bounded number of times. The complexity of the reachability problem for DAG $k$-tunnel gadgets, as well as the 2-player and team games, is characterized in \cite{gadgets2}.

\begin{definition}
	A gadget is \defn{deterministic} if every traversal can put it in only one state and every location has at most
	1 traversal from it.
	More precisely, its transition graph has maximum out-degree 1.
\end{definition}

\begin{definition}
	A gadget is \defn{reversible} if every transition can be reversed. More precisely, its transition graph is undirected.
\end{definition}

Reversible deterministic gadgets are gadgets whose transition graphs are partial matchings, and they naturally lead to unbounded problems. \cite{gadgets2} characterizes the complexity of reachability for reversible deterministic $k$-tunnel gadgets and partially characterizes the complexity of the 2-player and team games.

We define the decision problems we consider in their corresponding sections.


\section{0-Player Motion Planning with Spawners}
\label{sec:spawners}

In this section, we describe a model of 0-player motion planning,
introduce the spawner gadget, and
show that 0-player motion planning with spawners is RE-complete, implying undecidability. RE-completeness is defined in terms of arbitrary computable many-one reductions; in particular, they don't have to run in polynomial time. We will use the fact that the halting problem for 3-counter machines is RE-complete \cite{counter}.

\subsection{Model}

In \defn{0-player directed-edge motion planning} (with one robot),
we modify 1-player motion planning by
removing the player's ability to control the robot, and
specifying directions on the connections between gadget locations.
More precisely, the connection graph is now a directed graph
such that each gadget location has only incoming edges
(meaning that the robot enters the gadget from that location),
or only outgoing edges and at most one such edge
(meaning that the robot exits the gadget from that location);
and all gadgets must be deterministic.\footnote{There was no need to apply directions to the connection graph in \cite{ani2022trains} because each location acted exclusively as either the start of transitions or the end of transitions. In \cite{demaine2022pspace} the connections were undirected and it was assumed the robot proceeded away from the gadget where it just traversed.}
Thus the robot moves on its own,
moving in the direction of the edge it is on
and traversing any gadgets it encounters.
The reachability question asks whether the robot
reaches a specified target location in finite time.

Because the state of this system can be encoded in a polynomial number of
bits (the state for each gadget and the location of the robot),
this reachability problem is in PSPACE as in other 0-player models of
the gadget framework \cite{ani2022trains,demaine2022pspace}.

Our extension is to define the \defn{spawner} gadget:
a 1-location gadget that spawns a new robot in each round,
appearing at its only location.
We now define 0-player directed-edge motion planning to take into
account multiple robots and spawners.
\defn{0-player directed-edge motion planning with spawners}
is divided into rounds.
In each round, each robot takes a turn in spawn order,
and then each spawner spawns a robot (in a predefined spawning order).
A robot's turn consists of it moving along the directed edge it is on
until it either traverses a gadget or it gets stuck
(i.e., reaches a point where all edges are directed to its position).
The reachability question asks whether any robot
reaches a specified target location in finite time.

\begin{lemma}
	Deciding robot reachability in 0-player directed-edge motion planning with spawners with any set of gadgets is in RE.
\end{lemma}
\begin{proof}
	After each step of the game, there will still be a finite, if increasing, number of robots. Thus to confirm if at least $1$ robot can reach the win location in finite time we can simply simulate the game for the needed finite number of steps.
\end{proof}

\subsection{RE-hardness}

We show that deciding robot reachability in 0-player directed-edge motion planning with spawners is RE-hard by reduction from the 
halting problem by simulating a 3-counter machine. 
First we introduce the gadgets that we show RE-hard.

\paragraph{Increment gadget.}
The \defn{increment gadget} is a 4-state 10-location gadget containing a 3-path \defn{lock branch} and a 3-path \defn{path selector}
(Figure~\ref{fig:Increment}).
When a robot traverses a path in the path selector, it enables a single path in the lock branch and locks the path selector.
When a robot traverses a path in the lock branch, the gadget reverts to the original state.
\begin{figure}
	\centering
	\includegraphics[width=.7\linewidth]{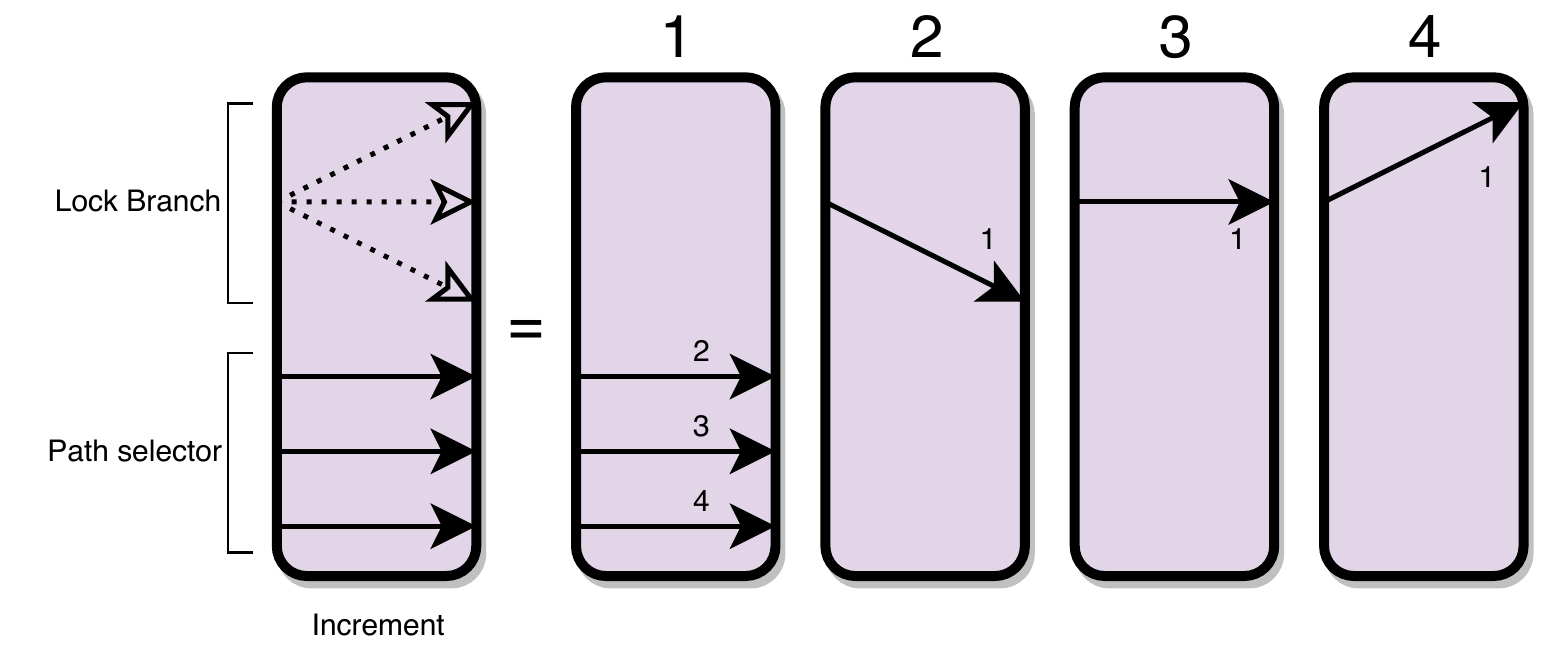}
	\caption{The increment gadget, shown with state transitions.}
	\label{fig:Increment}
\end{figure}

\paragraph{Register gadget.}
The \defn{register gadget} is a 3-state 10-location gadget containing a \defn{path selector}, a \defn{processing branch}, and a
\defn{response branch} (Figure~\ref{fig:Register}).
When a robot traverses the top path selector path, the path selector is locked and a path in the processing
branch is enabled. When a robot traverses the bottom path selector path, the path selector is locked and
the other processing branch path and a path in the response branch are enabled. If a robot traverses any non-path-selector path,
the gadget reverts to the original state.
\begin{figure}
	\centering
	\includegraphics[width=.7\linewidth]{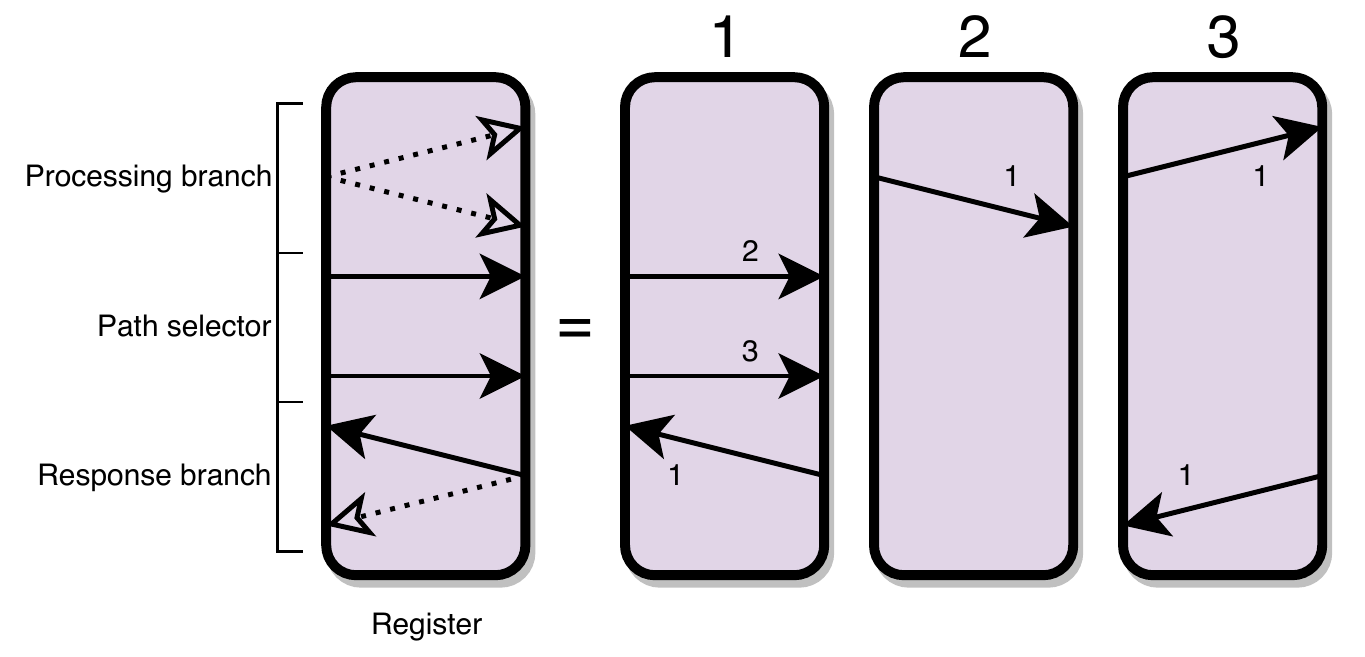}
	\caption{The register gadget, shown with state transitions.}
	\label{fig:Register}
\end{figure}

\paragraph{UPDSDS gadget.}
For the following theorem, we will also use the \defn{UPDSDS} gadget.
This gadget has two states ``up'' and ``down'', a tunnel which sets the state to ``up'', and two \defn{set-up switches} which each have one input and two outputs, where the output taken depends on the state and traversing the switch sets the state to ``down''.

\begin{theorem}
	Deciding robot reachability for 0-player directed-edge motion planning with spawners is RE-hard with the spawner, increment, register,
	and UPDSDS gadgets combined.
\end{theorem}
\begin{proof}
	We reduce from the halting problem of the 3-counter machine with \texttt{INC(r)}, \texttt{DEC(r)}, and
	\texttt{JZ(r, z)} instructions, which is undecidable (\cite{counter}).
	We will need to implement the \texttt{INC(r)} (increment register $r$ by 1), \texttt{DEC(r)} (decrement $r$ by 1), and
	\texttt{JZ(r, z)} (jump to instruction $z$ if $r$ is 0) instructions of a counter machine. We will not worry about
	decrementing a register that is already 0, because all \texttt{DEC} instructions can be preceded by \texttt{JZ} to guard against that.
	We will also implement the \texttt{HALT} instruction, which should result in a win.
	
	First we implement a \defn{register}, which will store a nonnegative integer, just like a register in a counter machine. This,
	of course, uses the register gadget, and the implementation is shown in Figure~\ref{fig:Register-impl}.
	In this implementation, the value of a register gadget is the number of robots stuck
	at the entrance of the processing branch. If a robot $b$ crosses the \emph{decrement in} path, a single robot can cross the gadget to
	the sink, where it is stuck forever,
	and all other robots stuck at the entrance stay stuck. Robot $b$ goes through the \emph{out} path on its next turn.
	This decrements the value of the gadget by $1$, thus implementing
	\texttt{DEC}, taking 1 round to process. If a robot $b$ crosses the \emph{jump-zero in} path, then if the gadget's value is
	nonzero, a single robot $b'$ crosses the top path of the processing branch,
	reverting the gadget's state, and forcing $b$ to traverse the top path of
	the response branch on its next turn, which leads to the \emph{out} path. $b'$ gets stuck back at the entrance on its next turn.
	However, if the gadget's value is $0$, then no robot will traverse the processing branch, which lets $b$ traverse the bottom path of
	the response branch on its next turn. This does not change the value of the gadget, and changes the path of $b$ iff
	the value is $0$, thus implementing \texttt{JZ}, taking 2 rounds to process.
	\begin{figure}
		\centering
		\includegraphics[width=.7\linewidth]{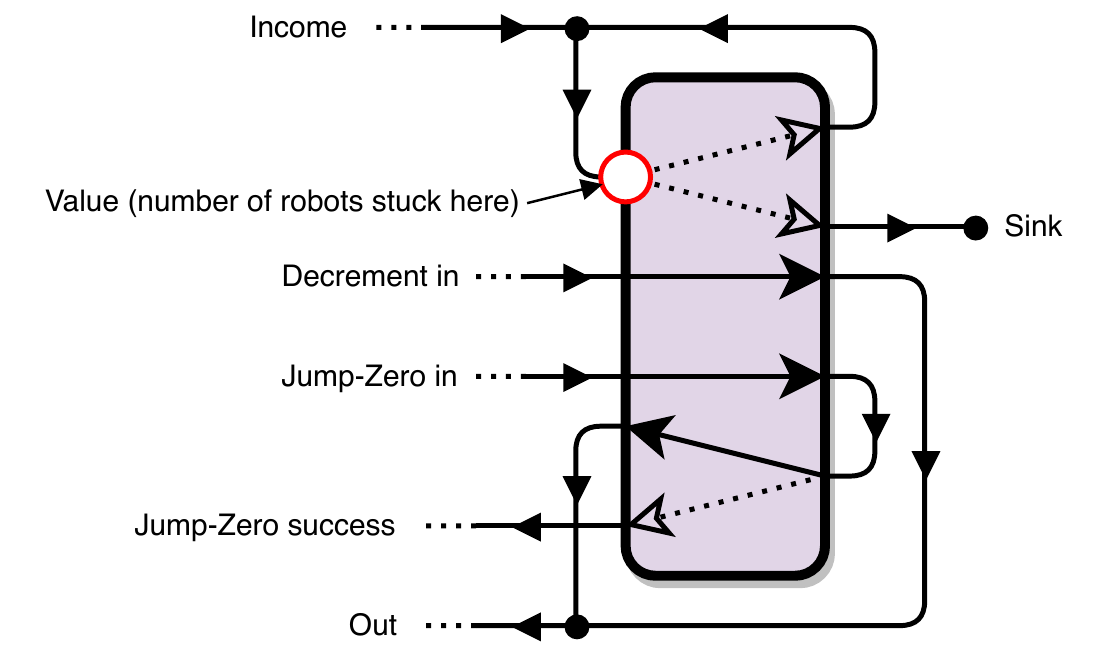}
		\caption{Implementation of the register of a counter machine.}
		\label{fig:Register-impl}
	\end{figure}
	
	To implement \texttt{INC}, we need a place that robots can come from. For this, we have the setup shown in Figure~\ref{fig:Increment-impl}.
	This setup contains a spawner gadget. Spawned robots go through
	the US gadget (a set-up switch, simulated by using one switch of the UPDSDS gadget and flipping it)
	to the entrance of the lock branch of the increment gadget and get stuck. It takes $2$ turns for this to happen. The first
	robot $b$ to get spawned instead takes the bottom path of the US gadget and executes the program. So during the 4th and later
	rounds, an extra robot gets stuck at the increment gadget. When robot $b$ goes through the \emph{increment $r_i$ in} path,
	a single robot $b'$ at the increment gadget traverses the lock branch, goes to the \emph{income} entrance of $r_i$, and gets stuck at that
	register gadget's processing branch on its next turn, incrementing said register gadget's value.
	In the process, the increment gadget reverts to its original state. This implements \texttt{INC}, taking 2 rounds to process,
	and we only need to make sure
	that $b$ does not traverse the path selector of the increment gadget before the 4th round to ensure that there will
	be a robot $b'$ that goes to a register.
	\begin{figure}
		\centering
		\includegraphics[width=.7\linewidth]{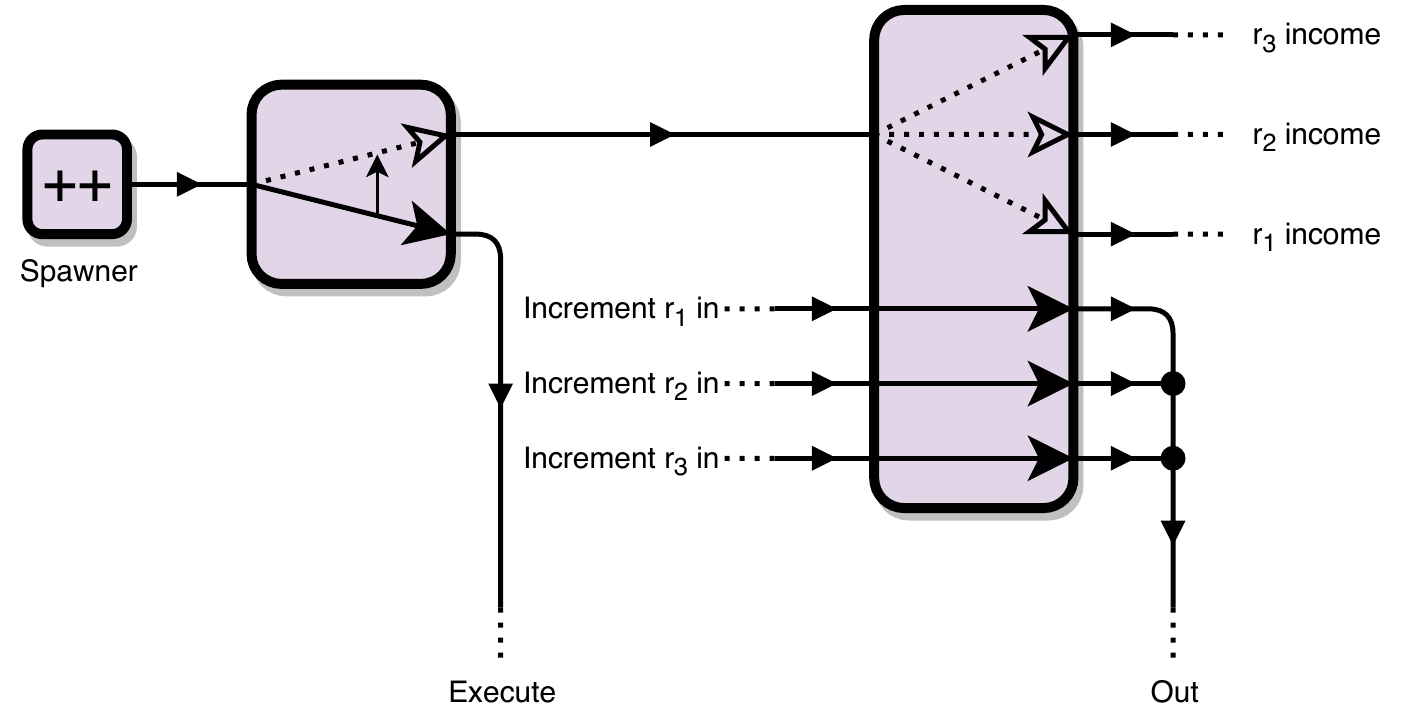}
		\caption{The context of the increment gadget, along with the spawner and a US gadget.}
		\label{fig:Increment-impl}
	\end{figure}
	
	We also need to implement the program, and we use UPDSDS gadgets for that, as shown in Figure~\ref{fig:UPDSDS-impl}. A UPDSDS-gadget
	instruction contains an \emph{execute in} entrance, a \emph{pass in} entrance, a \emph{jump in} entrance,
	a \emph{jump destination} entrance, an \emph{execute out} exit, an \emph{execute next} exit, a \emph{pass next} exit,
	a \emph{jump next} exit, and a \emph{jump out} exit.
	Only the executor robot is allowed to traverse this gadget.
	
	The \emph{execute out} exit leads to the proper location of the
	increment or register gadgets. For an \texttt{INC(r)} instruction, it leads to the \emph{increment $r$ in} entrance of the increment gadget.
	For a \texttt{DEC(r)} instruction, it leads to the \emph{decrement in} entrance of the register gadget for register $r$. For a
	\texttt{JZ(r, z)} instruction, it leads to the \emph{jump-zero in} entrance of the register gadget for register $r$. For
	a \texttt{HALT} instruction, it leads directly to the win location.
	
	The \emph{execute next} exit leads to the \emph{execute in} entrance of the next instruction. The \emph{pass next} exit leads to
	the \emph{pass in} entrance of the next instruction. The \emph{jump out} exit leads to the \emph{jump destination} entrance of
	instruction $z$ for a \texttt{JZ(r, z)} gadget, and doesn't exist otherwise. The \emph{jump next} exit leads to the
	\emph{jump in} entrance of the next instruction.
	\begin{figure}
		\centering
		\includegraphics[width=.7\linewidth]{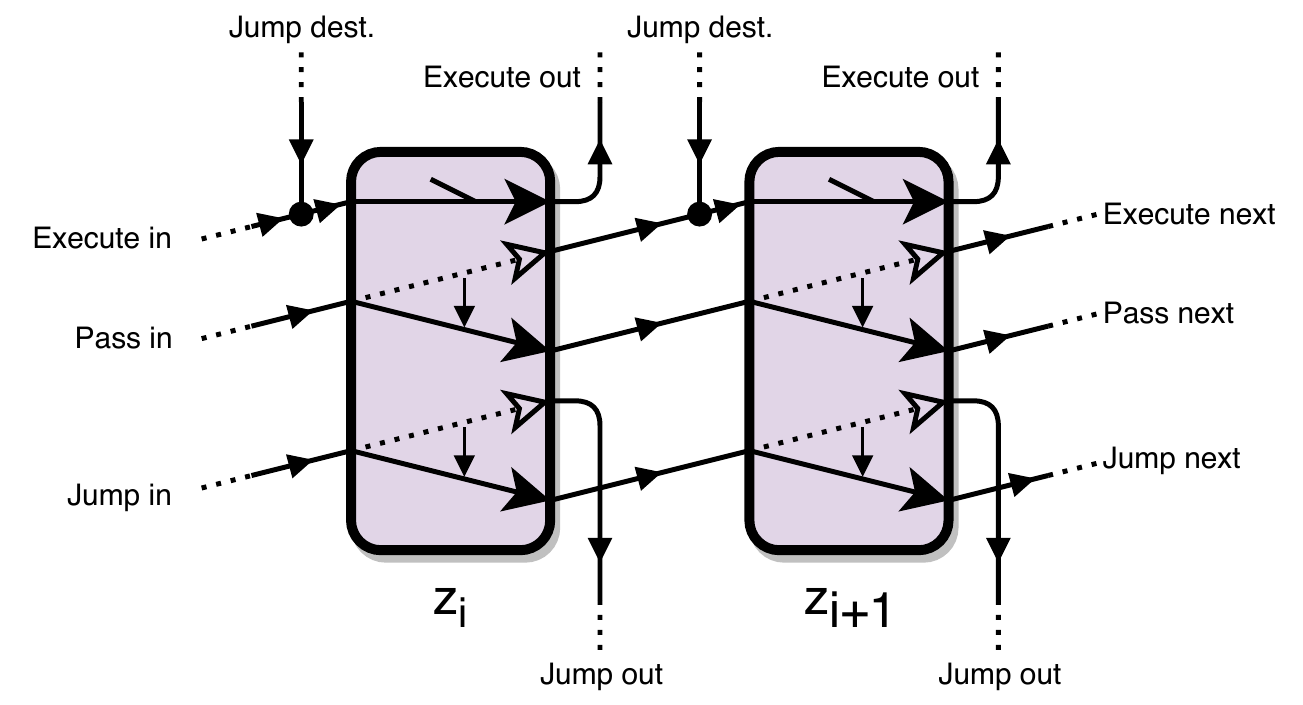}
		\caption{Two instructions implemented using UPDSDS gadgets.}
		\label{fig:UPDSDS-impl}
	\end{figure}
	
	This reduction can be done in polynomial time with respect to the number of instructions, because each instruction is
	simulated with 1 UPDSDS gadget, and there are a constant number of constant-size gadgets other than these.
	
	We now describe the behavior of the entire simulation, with an example shown in Figure~\ref{fig:counter}.
	\begin{itemize}
		\item A robot spawns from the spawner.
		\item The robot that spawned takes the bottom path of the US gadget, setting it to the \emph{up} state permanently. This
		robot is the executor robot. Another robot spawns from the spawner.
		\item The executor robot takes the top path of the UPDSDS gadget representing the first instruction. The newly spawned robot
		crosses the US gadget. Another robot spawns from the spawner.
		\item If the executor robot is executing an \texttt{INC} instruction, it traverses the path selector of the increment gadget.
		This is the 4th (or later) round, so there will be a robot ready to traverse the lock branch of the increment gadget.
		\item When the executor robot finishes executing an instruction that doesn't lead to a jump, it travels along the upper
		set-down switches of the UPDSDS gadgets until it finds the one representing the instruction it was executing. It resets that
		gadget and executes the next instruction, flipping the state of the next UPDSDS gadget.
		\item If the instruction led to a jump instead, the executor robot travels along the lower set-down switches of the UPDSDS gadgets
		until it finds the one representing the instruction it was executing. It resets that gadget and takes the \emph{jump next} path
		to the destination UPDSDS gadget of the jump, then executes the corresponding instruction.
		\item If the executor robot reaches the top path of the UPDSDS gadget representing the \texttt{HALT} instruction, it
		goes to the win location.
	\end{itemize}
	
	\begin{figure}
		\centering
		\includegraphics[width=\linewidth]{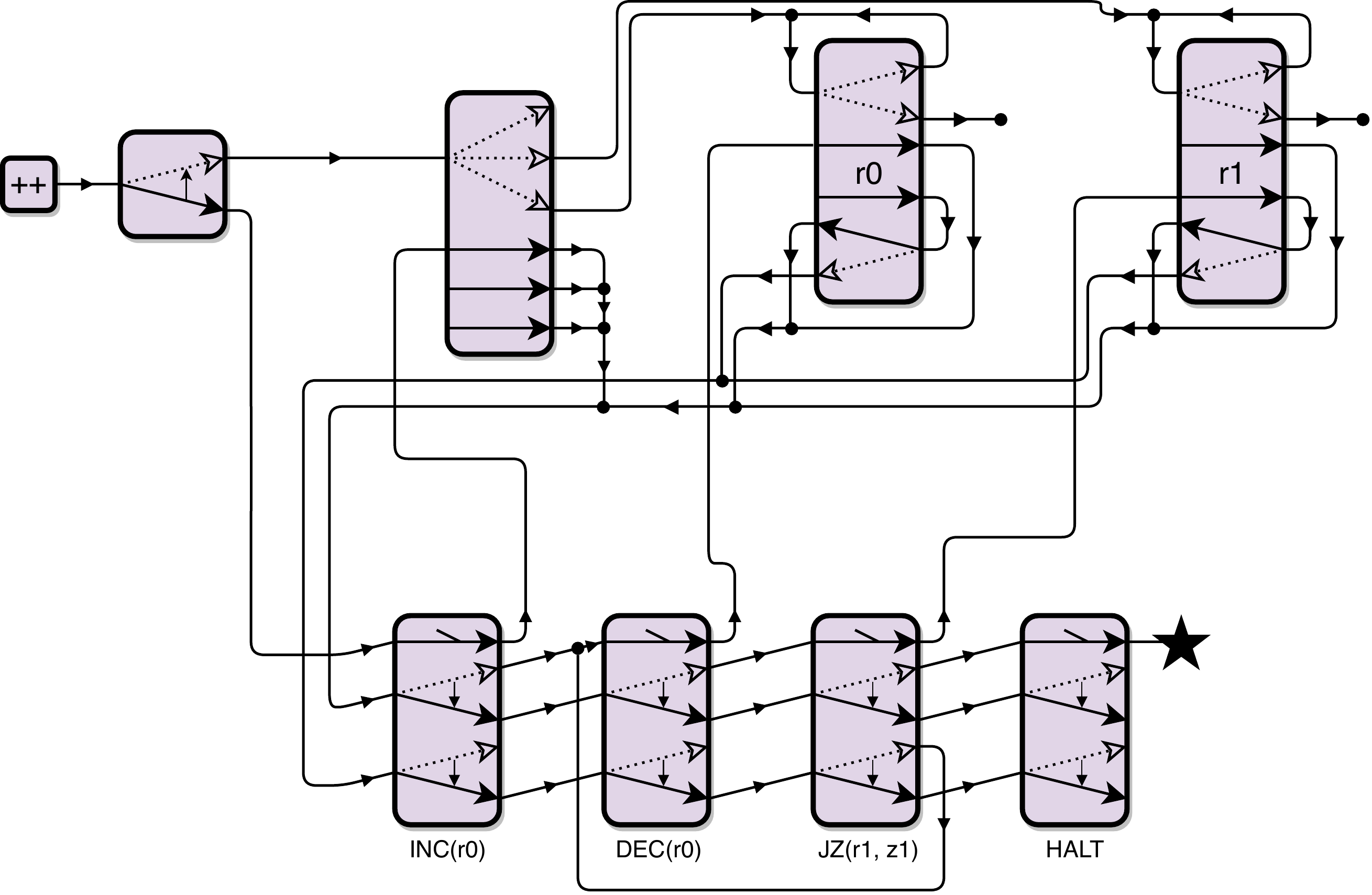}
		\caption{A 2-counter machine constructed with the gadgets. 2 counters are shown instead of 3 to save space.}
		\label{fig:counter}
	\end{figure}
	
	So this simulates a 3-counter machine. So if the 3-counter machine halts, then a robot will reach the win location in finite time,
	and vice versa.
\end{proof}


\section{1-Player Motion Planning with Spawners and/or Destroyers}
\label{sec:1playerInf}

In this section, we investigate 1-player motion planning with multiple robots,
where a single player controls a set of robots, with the ability to
separately command each, moving any one robot at a time. There is no limit to the number of robots that can be at a given location.
We include a \defn{spawner} gadget (as in Section~\ref{sec:spawners})
which the player can use to produce a new robot at a specific location,
providing an unlimited source of robots at that location.
We optionally also include a \defn{destroyer} gadget,
which deletes any robot that reaches a specified sink location;
such removal plays a role when we consider the \defn{targeted reconfiguration} problem
where the goal is to achieve an exact pattern of robots at the locations. If a system of gadgets only has a single spawner gadget we call that gadget the \defn{source} and if the system only has a single destroyer gadget we call that the \defn{sink}.

We show an equivalence between this 1-player motion planning
problem and corresponding problems on Petri nets.
Through these connections, we establish EXPSPACE-completeness for reachability;
PSPACE-completeness for reconfiguration with a spawner;
and ACKERMANN-completeness for reconfiguration with a spawner and a
destroyer.

\subsection{Petri Nets}
Petri nets are used to model distributed systems using tokens divided into dishes, and rules which define possible interactions between dishes. This is a natural model since many equivalent models have been defined such as Vector Addition Systems and Chemical Reaction Networks.


\begin{definition}
A \defn{Petri net} $\{D, R\}$ consists of a set of dishes $D$ and rules $R$.  A configuration $t$ is a vector over the elements of $D$ which represents the number of tokens in each dish. Each rule $(u, v) \in R$ is a pair of vectors over $D$. A rule can be applied to a configuration $d_0$ if $d_0 - u$ contains no negative integers to change the configuration to $d_1 = d_0 - u + v$. The volume of a configuration denoted $|d|$ is the sum of all its elements. 

\end{definition}

\begin{definition}
A reachable set for a Petri-net configuration, denoted $REACH_P(\{D, R\}, t)$, is the set of configurations of a Petri net reachable starting in configuration $t$ and applying rules from $R$.  
\end{definition}

We can view a system of gadgets with multiple robots as a set of gadget states $\Gamma$ and a vector $l$ indicating the counts of robots at each location. We can define the set of reachable targeted configurations as $REACH(\Gamma, l)$ similarity to Petri nets. 

\begin{figure}
	\centering
	\includegraphics[width=0.3\linewidth]{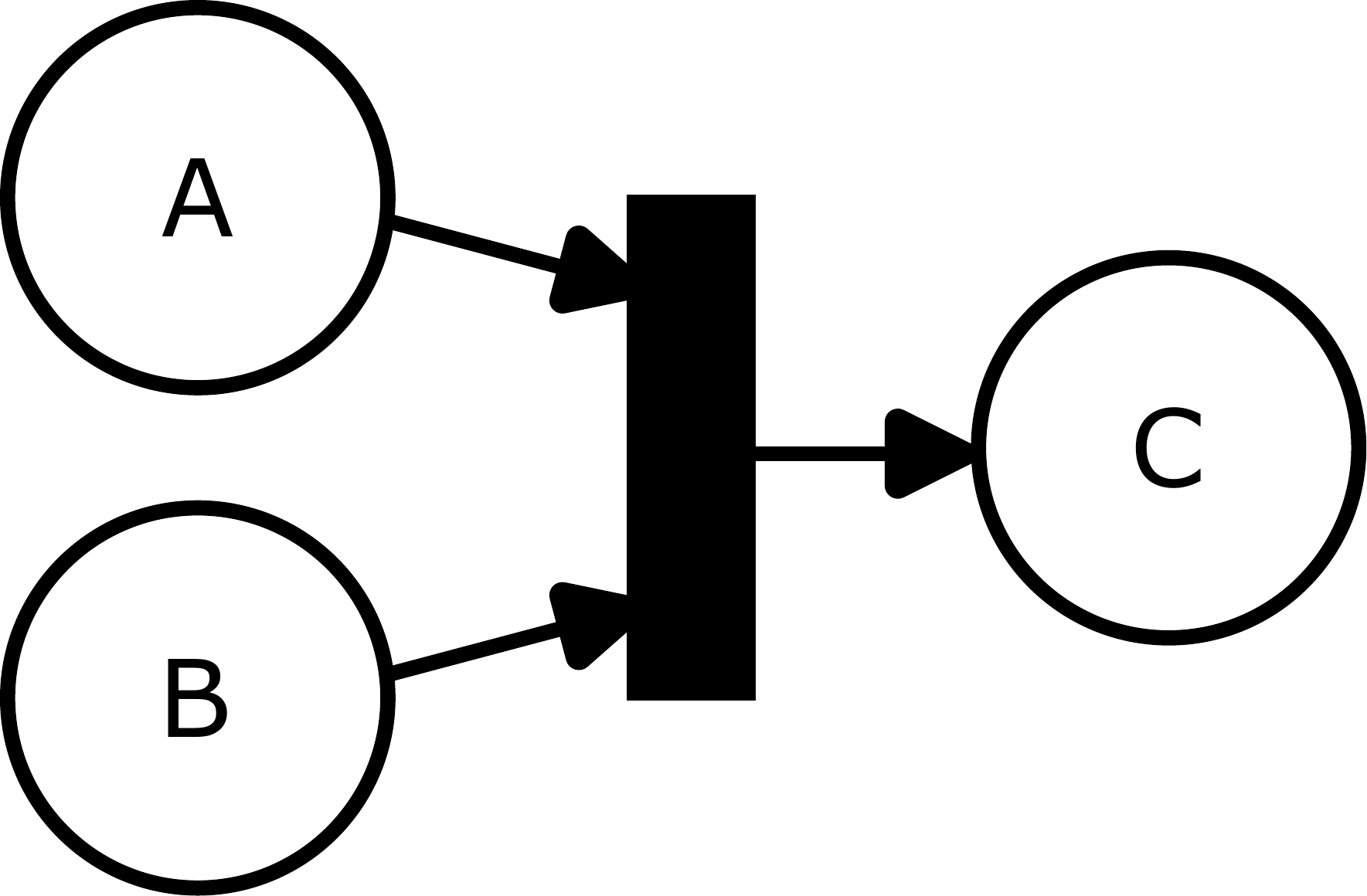}
	\caption{General Petri-net rule $(u,v)$, where $u$'s nonzero dishes are shown on the left side and $v$'s nonzero dishes are shown on the right side.}
	\label{fig:petriNet}
\end{figure}


\subsection{Equivalence between Petri Nets and Gadgets}
We present transformations that turn Petri nets into gadgets, and gadgets into Petri nets. We use these simulations to prove the complexity of robot reachability and reconfiguration with arbitrarily many robots.

\paragraph{Gadgets to Petri Nets.} 
We can transform a set of gadgets into a Petri net where each location, besides the source and sink, is represented as a \emph{robot dish}. Each gadget besides the spawner and destroyer is given a number of \emph{state dishes} equal to its states, and each transition of the gadget is represented by a \emph{rule}. The set of dishes $D$ is $D_{STATE} \cup D_{LOCT}$, the union of state and robot dish sets, respectively.

A configuration of robots and gadgets is represented by a Petri-net configuration $t$ satisfying the following:
\begin{itemize}
	\item Each $k$-state gadget is simulated by $k$ unique dishes in $D_{STATE}$, one per state. The state of the gadget is represented by a single token which is contained in the corresponding dish, and the other $k-1$ dishes are empty.
	\item Each location in the system of gadgets is simulated by a unique dish in $D_{LOCT}$. The number of tokens in that dish is equal to the number of robots at that location. 
\end{itemize}

A Petri net $\{D, R\}$ simulates a system of gadgets $G$ if for any configuration $\{\Gamma, l\}$ of $G$ represented by Petri-net configuration $t$, each configuration in $REACH_G(\Gamma, I)$ is represented by a configuration $REACH_P(\{D, R\}, t)$ and each configuration in $REACH_P(\{D, R\}, t)$ represents a configuration in $REACH_G(\Gamma, I)$.

\begin{figure}
	\centering
	\includegraphics[width=0.4\linewidth]{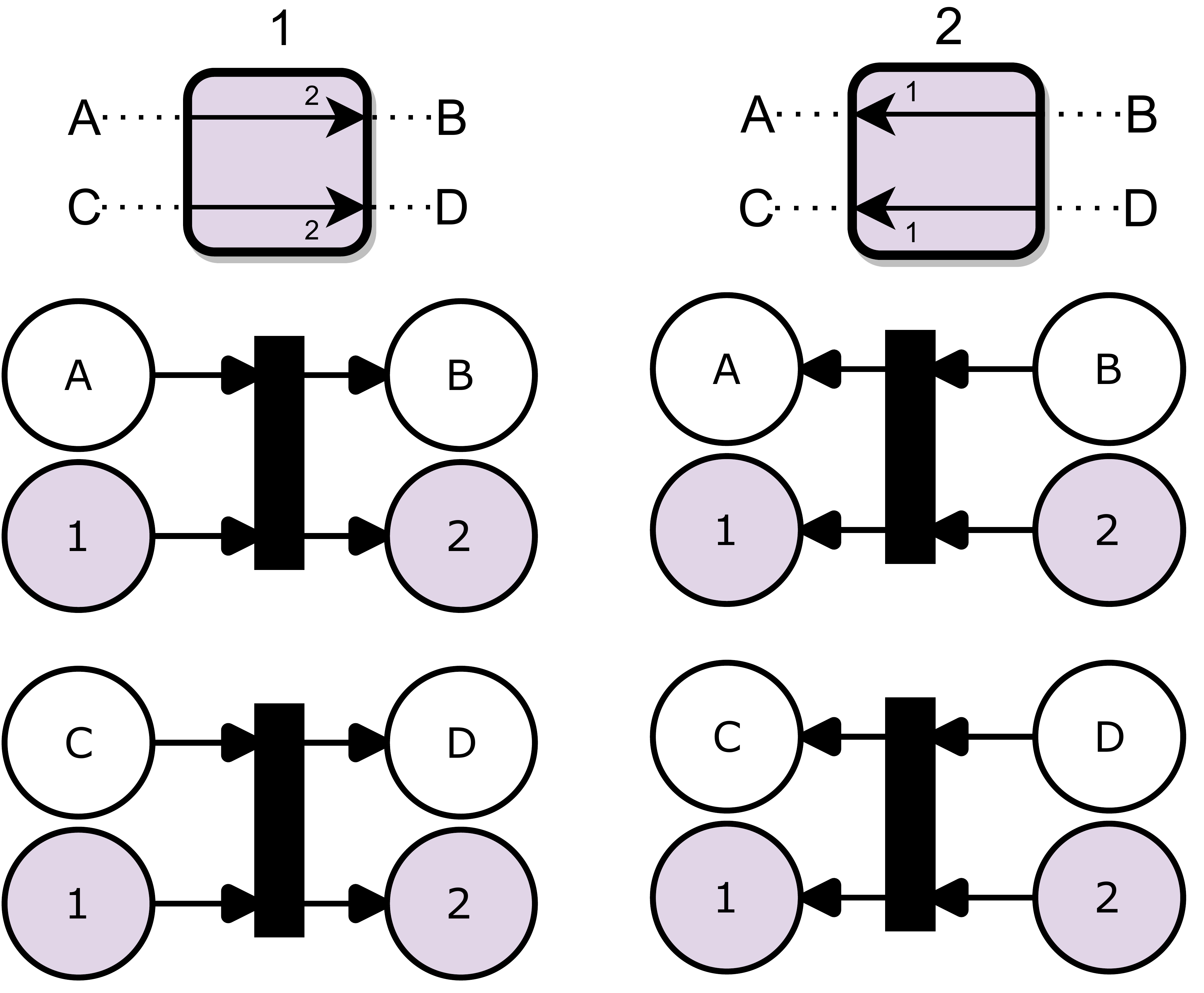}
	\caption{Petri-net rules which simulate a 2-tunnel toggle gadget.}
	\label{fig:gadgetTOpetri}
\end{figure}

\begin{lemma}\label{lem:gadgets2petri}
For any set of deterministic gadgets $S$, any system of multiple copies of gadgets in $S$ with a spawner (and optionally, a destroyer) can be simulated by a Petri net.
\end{lemma}
\begin{proof}

We first explain how to create the rules for gadgets that are not connected to the source or sink locations. Each gadget transition will be represented by a unique rule. For example the 2-tunnel toggle gadget is shown in Figure~\ref{fig:gadgetTOpetri} and has four transitions. It can be traversed:
\begin{itemize}
	\item from $A$ to $B$ in state $1$,
	\item from $C$ to $D$ in state $1$,
	\item from $B$ to $A$ in state $2$, and
	\item from $D$ to $C$ in state $2$.
\end{itemize}

The four corresponding rules for the gadget are drawn in Figure~\ref{fig:gadgetTOpetri} as well. Each rule takes in one token from a robot dish and one from a state dish, and places one token in a robot dish and one in a state dish. The token being moved between robot dishes models moving one robot across a gadget, and the token being moved between state dishes models the state change of the gadget.

If a gadget is connected to the source, any transition from the source is represented by a rule that only takes in a state token, producing two tokens. One token is output to a location dish and one to a state dish. If a transition is connected to the sink then the rule takes in two tokens and outputs only a state token. These special cases are shown in Figure~\ref{fig:sinkSource}. Note that we do not have an actual dish for the source so the player may spawn multiple robots at the source but they do not appear in the simulation until they traverse a gadget.

\begin{figure}
	\centering
	\includegraphics[width=0.45\linewidth]{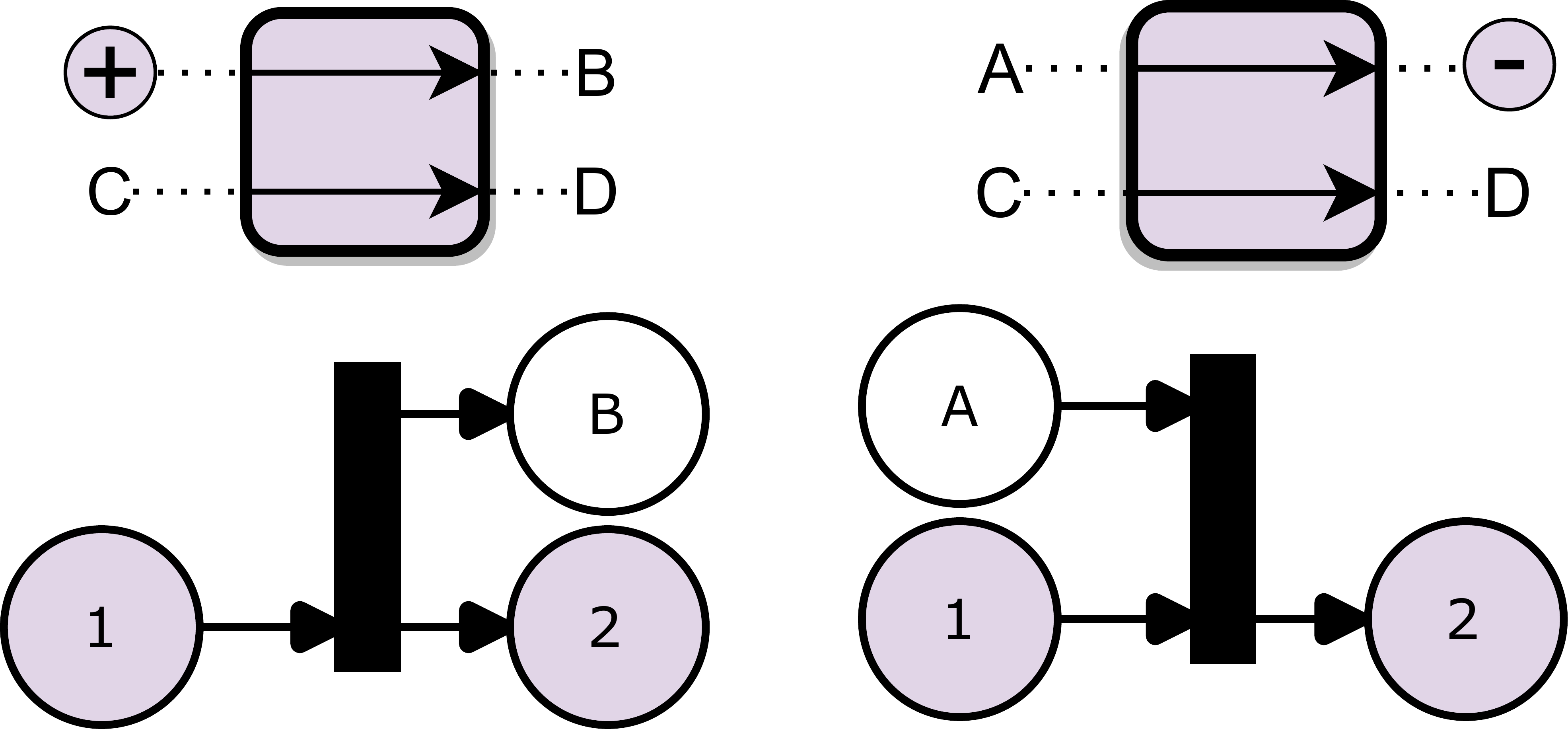}
	\caption{Left: Rule we include when a gadget can be traversed from the source. Right: Rule we include when a traversal leads to the sink.}
	\label{fig:sinkSource}
\end{figure}

For each configuration of a system of gadgets, there exists a configuration of the Petri net with dishes that represent the gadgets and locations. Each rule of the Petri net acts as a traversal of a robot changing the state of a gadget. The rules need the gadgets state token to be in the correct dish, and a robot token in the location dish representing the start traversal.
\end{proof}

\paragraph{Petri Nets to Gadgets.}
We simulate a Petri net with symmetric self-closing doors using a location for each dish, where each rule is represented by multiple gadgets. We also have a single \defn{control robot} which starts in a location we call the \defn{control room}. The other robots are \defn{token robots} which represent the tokens in each dish.  At a high level, our simulation works by ``consuming'' the input tokens to a rule to open a series of tunnels for the control robot to traverse. The control robot then opens a gadget for each output to allow token robots to traverse into their new dishes. We use the source and sink to increase and decrease rules as needed. Figure~\ref{fig:petriTOgadget} gives an overview. 

\begin{figure}
	\centering
	\includegraphics[width=0.5\linewidth]{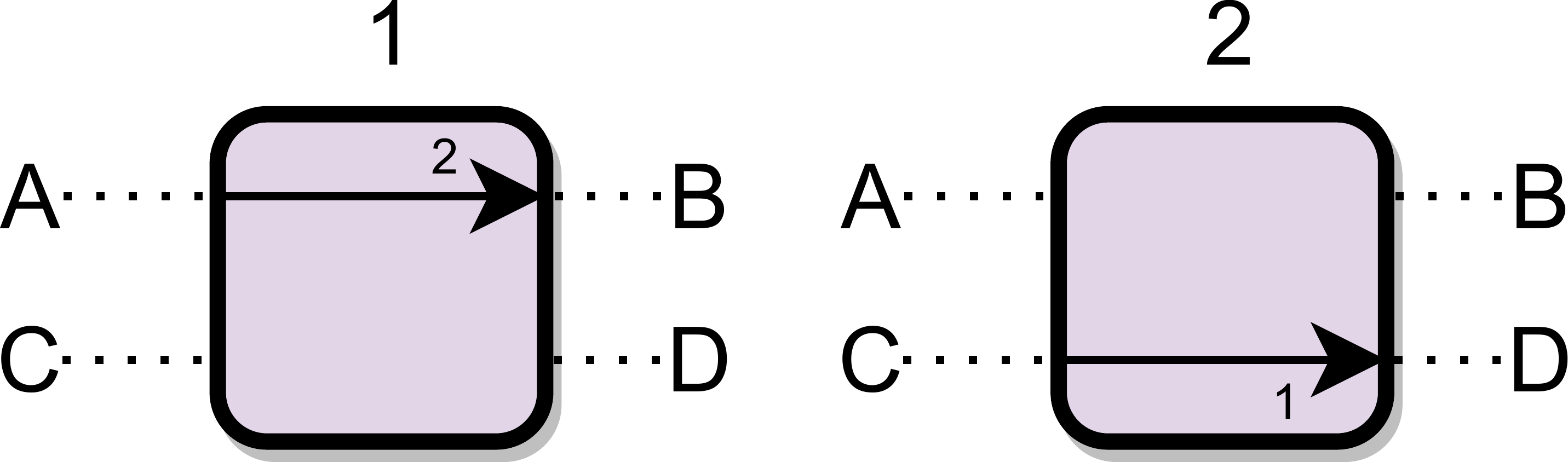}
	\caption{Symmetric self-closing door.}
	\label{fig:selfclosingDoor}
\end{figure}

\textbf{Symmetric self-closing door.} The \defn{symmetric self-closing door} is a $2$-state $2$-tunnel gadget shown in Figure~\ref{fig:selfclosingDoor}. The states are $\{1, 2\}$ and the traversals are 
\begin{itemize}
	\item in state $1$ from $A$ to $B$ changing state to $2$, and
	\item in state $2$ from $C$ to $D$ changing state to $1$. 
\end{itemize} 

\begin{figure}
	\centering
	\includegraphics[width=0.9\linewidth]{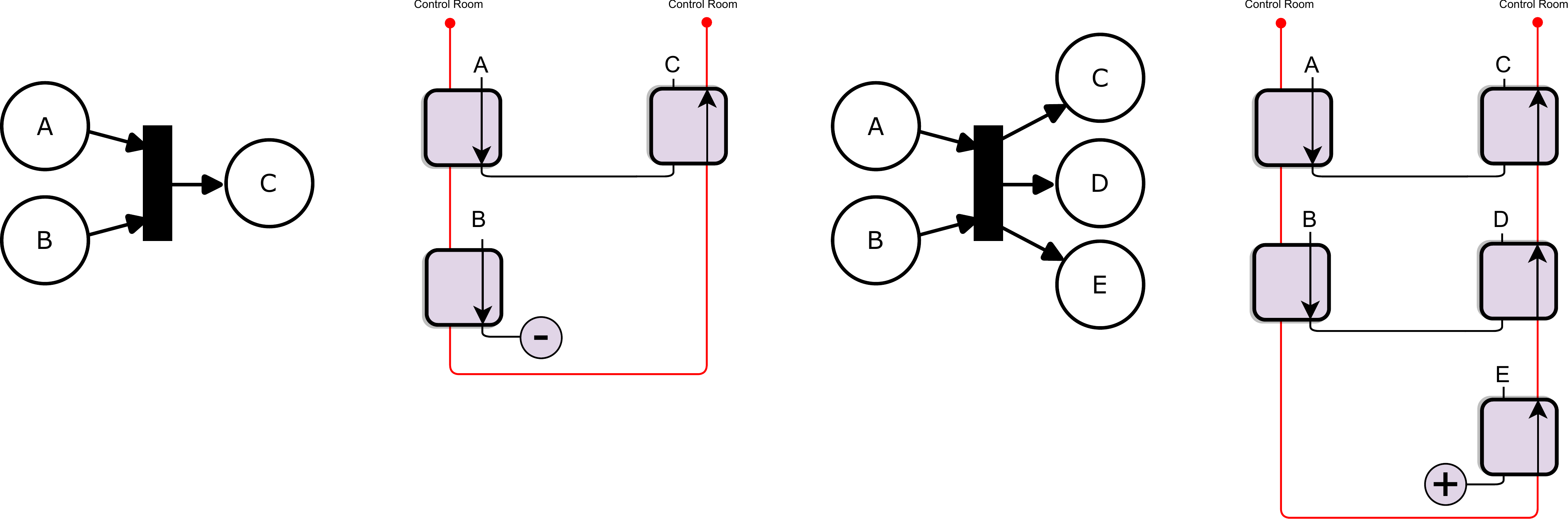}
	\caption{How to simulate a rule which decreases volume (Left) and a rule which increases volume (Right). }
	\label{fig:petriTOgadget}
\end{figure}

Using this simulation we prove two problems in Petri-nets are polynomial time reducible to the gadgets problems we are interested in. \cite{esparza2005decidability} lists many problems including the ones we describe here.\footnote{Problems names may differ.} First is production, this problem asks given a Petri-net configuration and a target dish, does there exist a reachable configuration which contains at least one token in the target dish.  Configuration reachability asks given an initial and target configuration, is the target reachable from the initial configuration.

\begin{lemma}\label{lem:Petri2Gadgets}
Production in Petri nets is polynomial time reducible to robot reachability with the symmetric self-closing door and a spawner. Configuration reachability in Petri nets is polynomial-time reducible to multi-robot targeted reconfiguration with the symmetric self-closing door and a spawner. 
\end{lemma}
\begin{proof}
 For a rule $(a, b)$ we include $|a| + |b|$ copies of the gadgets. There is a gadget for each input to the rule; these gadgets can be traversed from the location representing an input dish to an intermediate location, opening another tunnel for the control robot to traverse. The control robot must traverse all the input gadgets the goes through the tunnels of the output gadgets. The control robot opens the doors of these gadgets allowing the robots moving from an intermediate wire to traverse to a location representing the output dishes.

If a rule would increase the volume, the surplus output gadgets will allow traversal from the spawn location instead of an input gadget. If a rule decreases the volume, then the surplus input gadgets send robots to a ``sink'' location instead of an output gadget. We do not require a true sink in this case because we can add an extra location which robots can be held instead of being deleted. If we do not connect this location to any other gadget, then the robots can never leave and can be thought of as having left the system. 

Production reduces to robot reachability since a robot can reach a location if and only if a token can reach the corresponding dish. If token is placed in a dish, it must have moved through a rule gadget. The robot can only move through a rule gadget if the number of robots in the dishes are at least the number of tokens of the left hand side of the rules to open the tunnels for the control robot to move through. 

Configuration reachability in Petri nets reduces to multi-robot targeted reconfiguration. The target and initial states of the gadgets are the same. The only difference between the initial configuration and the target is the number of robots at each location, equal to the counts in the instance of Configuration reachability for Petri nets. The number of robots at each location is equal to the number of tokens in each dish. The targets for each intermediate wire is $0$ and in the control room $1$. Thus, it is never beneficial to partially traverse a rule gadget.
\end{proof}

\subsection{Complexity of Reachability}
The reachability problem for a single robot is very similar to the well-studied problem in Petri nets called coverage. The input to the coverage problem is a Petri net and a vector of required token amounts in each dish, and the output is yes if and only if there exists a rule application sequence to reach a configuration with at least the required number of tokens in each dish.

\begin{definition}[Coverage Problem]
\textbf{Input:} A Petri net $\{D, R\}$, and vectors $d_0$ and $d_c$.

\textbf{Output:} Does there exist a reachable configuration $d \in REACH(\{D, R\}, d_0)$ such that $d[k] \geq d_c[k]$ for all $0 \leq k < |D|$. 
\end{definition}

\begin{theorem}
Robot reachability is EXPSPACE-complete with symmetric self-closing doors, a spawner, and optionally a destroyer.
\end{theorem}
\begin{proof}
 We can solve robot reachability by converting the system of gadgets to a Petri net which simulates it as in Lemma~\ref{lem:gadgets2petri}. In this simulation, a token can be placed in a location dish if and only if a robot can reach that location represented by that dish. 
  Determining if a single token can be placed in a target dish, the production problem, is a special case of coverage problem where the target dish is labeled with $1$ and all others labeled with $0$.  
  We can use the exponential-space algorithm for Petri-net coverage shown in \cite{rackoff1978covering} to solve robot reachability. When simulating the sink we require rules that decrease the volume of a Petri net. This algorithm works for general Petri nets so it implies membership with a sink. 

For hardness, we first reduce Petri-net coverage to Petri-net production by adding a target dish $T$ starting with $0$ tokens and a new rule. This rule takes as input the number of tokens equal to the goal of the coverage problem and produces one token to the $t$ dish. 
This token can only produced if the reach a configuration that has at least the target number of each species.  We then use Lemma \ref{lem:Petri2Gadgets} to reduce production to robot reachability with the self-closing symmetric door and a spawner.  It is relevant to note the first reduction does not work when exactly the target numbers are required. 
The reduction works even when not allowing the sink as described in Lemma \ref{lem:Petri2Gadgets}. 
\end{proof}

\subsection{Complexity of Reconfiguration}
The reconfiguration problem has been studied in the single-robot case as the problem of moving the robot through the system of gadgets so that each gadget is in a desired final state. Targeted reconfiguration not only asked about the final states of the gadgets, but the location of the robot as well. Here, we study multi-robot targeted reconfiguration which requires both that all gadgets are in specified final states and that each location contains a target number of robots. 

\begin{definition}
For a gadget $G$, the \defn{multi-robot targeted reconfiguration problem for $G$} is the following decision problem. Given a system of gadgets consisting of copies of $G$ and the starting location(s)
a target configuration of gadgets and robots, 
is there a sequence of moves the robots can take to reach the target configuration? 
\end{definition}

The complexity of multi-robot targeted reconfiguration depends on whether we allow a destroyer. If we do not allow for a destroyer, the complexity is bounded by polynomial space since we can never have more robots than the total target size. If we allow for the ability to destroy robots, then the reconfiguration problem is the same as the configuration reachability problem in Petri nets from our relations between the models above. This is a fundamental problem about Petri nets and was only recently shown to be ACKERMANN-complete \cite{leroux2022reachability, czerwinski2022reachability}.

\begin{theorem}
Multi-robot targeted reconfiguration is ACKERMANN-complete with symmetric self-closing doors, a spawner, and a destroyer.
\end{theorem}
\begin{proof}
For membership we can solve multi-robot target reconfiguration by converting the gadgets to the Petri net using Lemma~\ref{lem:gadgets2petri}. The target configuration is a state token for each gadget in the dish of its target state, and a number of tokens in each location dish as the number of robots in the target configuration. We can then call the ACKERMANN algorithm for configuration reachability in Petri nets shown in \cite{leroux2019reachability}. 

For hardness we can reduce from configuration reachability.  It was shown in \cite{czerwinski2022reachability} that configuration reachability is ACKERMANN-hard. 
\end{proof}

The reduction presented in \cite{czerwinski2022reachability} vitally uses the ability of Petri nets to delete tokens, so we must use a sink in our simulation. Without a sink, we have PSPACE-completeness for multi-robot targeted reconfiguration. 

\begin{theorem}
Multi-robot targeted reconfiguration for symmetric self-closing doors and a spawner is PSPACE-complete.
\end{theorem}
\begin{proof}
Consider the input to the reconfiguration problem: two configurations of a system of gadgets. Namely, the start and end state of all the gadgets, and a start and end integer for each location. Since we can never destroy a robot once it is spawned, it always exists, so the player cannot spawn more robots than the total number of robots in the target configuration. We can then solve this problem in NPSPACE by nondeterministically selecting a robot to move, either from the source or another location. If we ever increase the total number of robots above the target we may reject. If we ever reach the configuration with the correct gadget states and robots at each location accept. Since PSPACE = NPSPACE we get membership. 

We inherit hardness from the 1-player single-robot case by not including the source or connecting it to an unreachable location. 
\end{proof}


\section{Impartial Unbounded 2-Player Motion Planning}
\label{sec:impartial}

In this section, we describe the 2-player impartial motion planning game and show that it is EXPTIME-complete for any
reversible deterministic gadget.

\subsection{Model}

In the \defn{2-player impartial motion planning game}, two players control the same robot in a system of gadgets. Player 1 moves first, then
Player 2 moves, then play repeats. On a given player's turn, they move the robot arbitrarily along the connection graph and through exactly one transition of a gadget.
There is also a \defn{ko rule}: The robot cannot traverse the same gadget on a player's turn as it traversed on their opponent's previous turn.
If a player cannot make the robot traverse a gadget without breaking the ko rule, that player loses and the other player wins.

\begin{lemma}\label{lem:2-player-in}
	Deciding whether Player 1 has a deterministic winning strategy in the 2-player impartial motion planning game is in EXPTIME
	for any set of gadgets.
\end{lemma}

\begin{proof}
	An alternating Turing machine can solve the problem by using existential states to guess Player 1's moves and universal states to guess Player 2's moves, accepting when Player 1 wins and rejecting when Player 2 wins. This takes only polynomial space because the configuration of the game can be described in polynomial space. The machine can reject after a number of turns at least the number of configurations, which is at most exponential and thus can be counted to in polynomial space. Hence the problem is in APSPACE${}={}$EXPTIME.
\end{proof}

\subsection{Hardness}

We introduce the \defn{locking 2-toggle}, introduced in \cite{gadgets2} and shown in Figure~\ref{fig:locking-2-toggle}. States 1 and 3 are
\defn{leaf states} and state 2 is the \defn{nonleaf state}. If a robot crosses a tunnel in state 2, the tunnel flips direction
and the other tunnel locks. Crossing a tunnel again will reverse this effect.

\begin{figure}
	\centering
	\includegraphics[width=.4\linewidth]{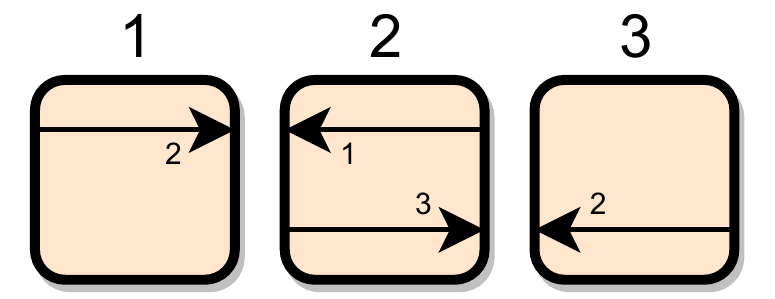}
	\caption{The locking 2-toggle.}
	\label{fig:locking-2-toggle}
\end{figure}

\begin{theorem}
	Deciding whether Player 1 has a deterministic winning strategy in the 2-player impartial motion planning game is EXPTIME-hard
	for the locking 2-toggle.
\end{theorem}
\begin{proof}
	We reduce from $G_4$ as defined in \cite{gx}. $G_4$ is a 2-player game involving Boolean variables where the players flip a variable
	on their turn and try to be the one to satisfy a common DNF Boolean formula with 13 variables per clause (a 13-DNF).
	Players have their own variables and can't flip their opponent's variables, and a player may flip 1 variable on their turn
	or pass their turn.	There is no ko rule.
	
	We start the robot next to a 1-toggle (a single tunnel of a locking 2-toggle) as shown in
	Figure~\ref{fig:2-player-start}. This 1-toggle is called the \defn{alternator}. On each side of the alternator is a variable
	system for each player, which consists of variable branching and variable setting loops. The \defn{variable branching}, as shown in Figure
	\ref{fig:2-player-var-branch}, has 2 locking 2-toggles before each branch. These start in the nonleaf state.
	At the end of each path is a \defn{variable flipping loop},
	which is shown in \ref{fig:2-player-var-set}.	
	The variable flipping loop for variable $v$ contains	2 locking 2-toggles per instance of $v$ or $\neg v$ in the 13-DNF formula of the
	$G_4$ instance, as well as an path to the 13-DNF checker with 2 1-toggles on it.
	The locking 2-toggles representing $v$ start in the nonleaf
	state iff $v$ starts True in $G_4$, and the locking 2-toggles representing $\neg x$ start in the leaf state iff $x$ starts
	True in $G_4$. One path of the variable branch, on the other hand, leads to a
	\defn{pass loop}, which is a variable flipping loop with 2 1-toggles in the loop instead of the locking 2-toggles.
	The 13-DNF checker contains a path
	for each clause in the 13-DNF, and each path contains a locking 2-toggle representing $v$, the same as one of the locking
	2-toggles representing $v$ in the variable flipping loop of $v$, followed by a 1-toggle, for each variable $v$ in the corresponding
	clause. The paths all lead to a final 1-toggle called the \defn{finish line}. This reduction can be done in polynomial time,
	as each variable and clause in $G_4$ is converted to a polynomial number of constant-size gadgets.
	
	\begin{figure}
		\centering
		\includegraphics[width=.7\linewidth]{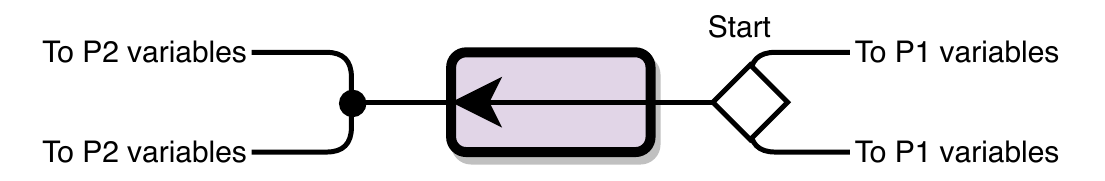}
		\caption{The robot's starting position, and the 1-toggle that's called the \emph{alternator}.}
		\label{fig:2-player-start}
	\end{figure}
	
	\begin{figure}
		\centering
		\includegraphics[width=\linewidth]{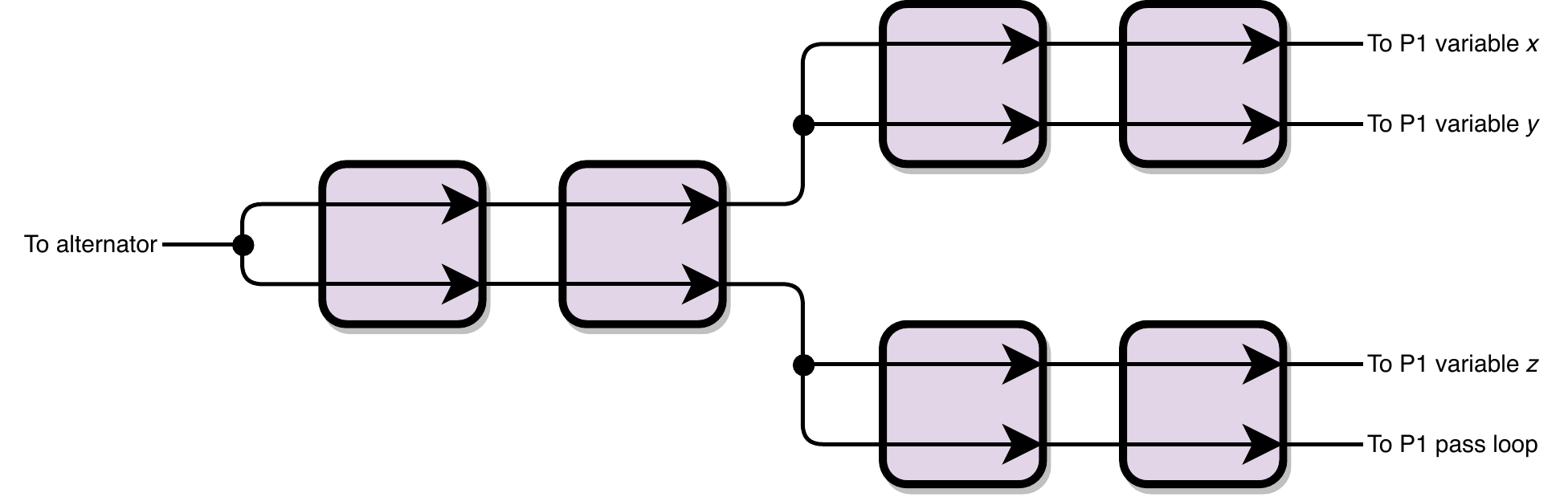}
		\caption{The variable branching for Player 1. Player 2's variable branching is on the other side of the alternator.
		In this example, player 1 has 3 variables: $x$, $y$, and $z$.}
		\label{fig:2-player-var-branch}
	\end{figure}
	
	\begin{figure}
		\centering
		\includegraphics[width=\linewidth]{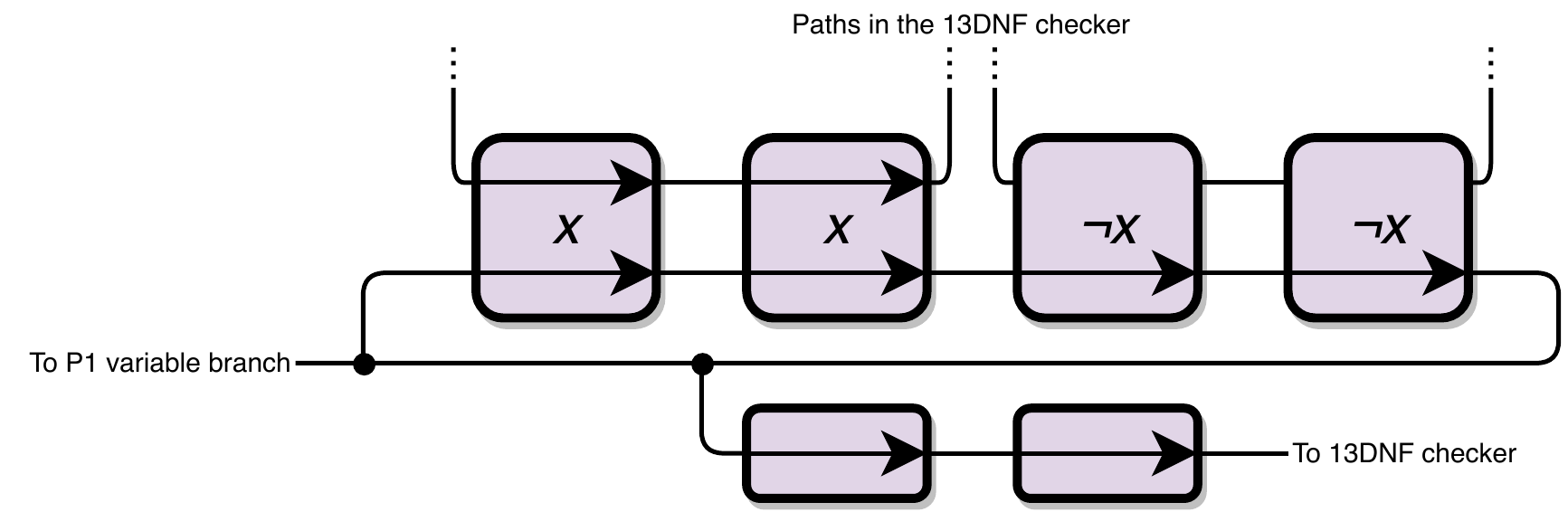}
		\caption{The variable flipping loop for variable $x$. This example represents the case where the 13-DNF has 1 instance of $x$
		and 1 instance of $\neg x$. Currently, $x$ is True.}
		\label{fig:2-player-var-set}
	\end{figure}
	
	\begin{figure}
		\centering
		\includegraphics[width=0.5\linewidth]{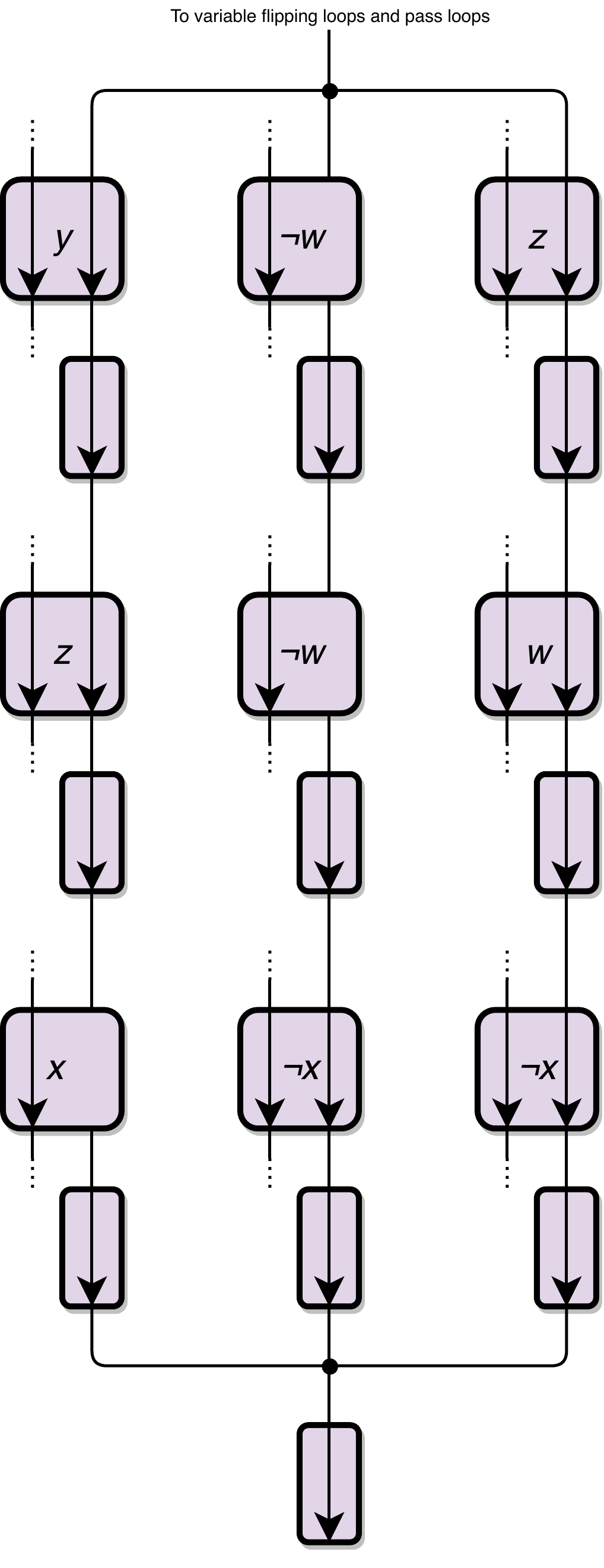}
		\caption{A 13-DNF checker, except that it represents a 3-DNF. This example represents
		$(y\vee z\vee x)\wedge(\neg w\vee\neg w\vee \neg x)\wedge(z\vee w\vee\neg x)$. The dotted paths are part of variable setting loops.}
		\label{fig:2-player-DNF-checker}
	\end{figure}
	
	During intended play:
	\begin{itemize}
		\item Player 1 moves the robot through variable branching to select a variable to set.
			Because the locking 2-toggles are doubled, and because
			of the ko rule, Player 2 has no choice but to second Player 1's choices. Player 1 could also move the robot to the pass loop.
		\item Player 1 moves the robot around a variable selection loop,
			a variable by flipping whether each locking 2-toggle is locked or not. If they're in
			the pass loop, they just go around the loop. Again, Player 2 has no choice since the number of gadgets in the path is even.
		\item Player 1 either moves the robot to the 13-DNF checker or back through the variable branching to the alternator.
		\item If Player 1 moves it back, they make it cross the alternator, and
			Player 2 goes through the same steps, but on the other side of the alternator.
		\item If a player moves the robot to the 13-DNF checker, they pick a path. If that path's corresponding clause in the 13-DNF is currently
			satisfied, they cross the finish line and win, since their opponent then has no legal moves. Otherwise, they get
			blocked by the first variable set to False, making their opponent win.
	\end{itemize}
	So Player 1 has the initiative and takes a $G_4$ turn on one side of the alternator, and Player 2 has the initiative
	and takes a $G_4$ turn on the other side. It is correct for a player to move the robot to the 13-DNF checker iff the 13-DNF
	is currently satisfied.
	
	We will now look at ways that the players can try to break the simulation of $G_4$:
	\begin{itemize}
		\item Player 1 can make the robot cross the alternator as their first move. However, this lets Player 2 flip a variable or pass first. If
			Player 1 can win this way, they can also win by passing (moving the robot around the
			pass loop) first and then giving the initiative to Player 2.
			So not crossing the alternator first is always a correct move.
		\item A player can move the robot to a variable flipping loop and cut to the 13-DNF checker. However, if the player can
			win this way, they can win by passing and moving the robot to the 13-DNF checker.
		\item A player can try to turn around and flip another variable on the way back to the alternator. However, the ko rule prevents this.
		\item A player can try to move the robot to some other variable flipping loop from the start of the 13-DNF checker. However,
			1-toggles will block the way.
	\end{itemize}
	Thus, the players are effectively forced to play $G_4$ in this game. Therefore, if Player 1 has a deterministic
	winning strategy in the $G_4$ instance, then they have one in this game, and if Player 1 has a deterministic winning strategy
	in this game, then they have one in the $G_4$ instance as well.
\end{proof}

\begin{theorem}\label{thm:2-player-hard}
	Deciding whether Player 1 has a deterministic winning strategy in the 2-player impartial motion planning game is EXPTIME-hard
	for \emph{any interacting $k$-tunnel reversible deterministic gadget}.
\end{theorem}
\begin{proof}
  Figure~\ref{fig:arbitrary-ikrdg} shows two tunnels that any
  interacting $k$-tunnel reversible deterministic gadget must have,
	as proved in \cite[Section~2.1]{gadgets2}, which further shows that
  these tunnels can be used to simulate a locking 2-toggle.
  For 2-player impartial motion
	planning, however, we must be careful of the simulation. To preserve parity, each traversal in the locking 2-toggle must correspond
	to an odd number of traversals in the simulation. In addition, if a traversal is not allowed, it must be blocked after an
	even number of traversals so the player who started moving the robot along that path loses. And to simulate the gadget ko rule,
	the gadgets at the ends of the simulation must be in the way of both paths. If all the constraints are met, then if a player
	makes the robot start a traversal along the simulation, the players must follow through, and in the end, it will be said player's
	opponent's turn. The opponent would have to make the robot traverse a gadget not in the simulation. Players would be disincentivized
	to start a traversal along a closed path, because they will be the one stuck with no legal moves. So the simulation would
	act exactly like a locking 2-toggle in the above reduction, giving us a straightforward reduction 2-player impartial motion planning
	with locking 2-toggles to 2-player impartial motion planning with any interacting $k$-tunnel reversible deterministic gadget.
	\begin{figure}
		\centering
		\includegraphics[width=.3\linewidth]{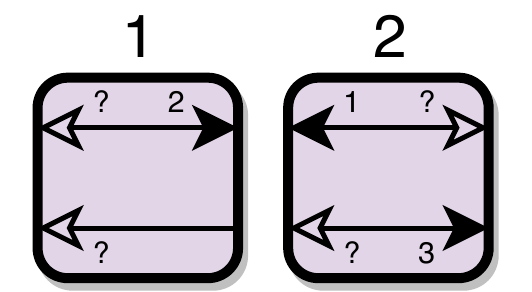}
		\caption{Two tunnels that an interacting $k$-tunnel reversible deterministic gadget must have. Solid arrows indicate open
		traversals, hollow arrows with ``?'' indicate optionally open traversals, and absent arrows indicate closed traversals.
		State 3 could be any state, including 1 and 2.}
		\label{fig:arbitrary-ikrdg}
	\end{figure}
	
	First we simulate a 1-tunnel reversible deterministic gadget with a directed tunnel, as shown in Figure~\ref{fig:directed-sim}.
	The robot cannot cross from right to left. If it crosses from left to right, it may cross back (after traversing some
	other gadget, of course), and the path from left to right
	may optionally still be open, this time leading to whatever state. Note that it takes two traversals to cross the simulation,
	and that a closed path in state 1 of the gadget used in the simulation blocks the robot after 0 traversals.
	\begin{figure}
		\centering
		\includegraphics[width=.5\linewidth]{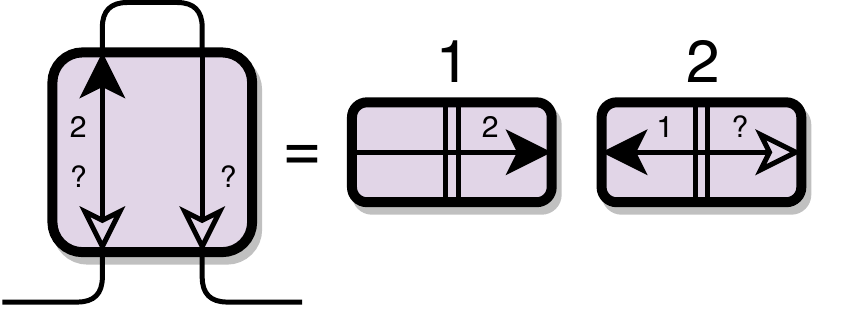}
		\caption{Simulation of a 1-tunnel reversible deterministic gadget with a directed tunnel. We draw double bars crossing the
		1-tunnel gadget as a reminder that it takes two traversals to cross.}
		\label{fig:directed-sim}
	\end{figure}
	
	Now we simulate the locking 2-toggle, as shown in Figure~\ref{fig:locking-2-toggle-sim}. The simulation currently simulates
	the locking 2-toggle in the nonleaf state. The robot can traverse from top right to top left or from bottom left to bottom right.
	The robot will get blocked after two traversals in an attempt to traverse from top left to top right or from bottom right to bottom left.
	If the robot traverses from top right to top left, the robot will be able to traverse from top left to top right (after traversing
	a different gadget). But an attempt to traverse from bottom left to bottom right gets the robot blocked after 0 traversals, thanks
	to the tunnel interaction in the left gadget, and an attempt to traverse from bottom right to bottom left or from top right
	to top left gets blocked after two traversals. So this would simulate a leaf state of the locking 2-toggle. The center gadget never
	becomes relevant for blocking, so we can argue by symmetry that traversing from bottom left to bottom right results in the other
	leaf state. Note that each path takes nine traversals to cross, so we have successfully simulated the locking 2-toggle meeting the
	constraints. This completes the proof.
	\begin{figure}
		\centering
		\includegraphics[width=\linewidth]{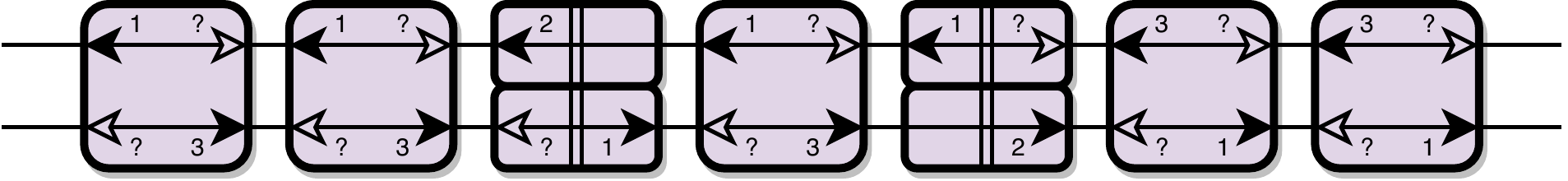}
		\caption{Simulation of the locking 2-toggle, under the constraints.}
		\label{fig:locking-2-toggle-sim}
	\end{figure}
\end{proof}

By Lemma~\ref{lem:2-player-in} and Theorem~\ref{thm:2-player-hard}, it is EXPTIME-complete to determine whether Player 1 has a deterministic
winning strategy in the 2-player impartial motion planning game with any interacting $k$-tunnel reversible deterministic gadget.


\section{Open Problems}
For 0-player motion planning, we leave as an open problem whether the finite-time reachability problem is undecidable for a smaller set of gadgets. In particular,
we used gadgets that can separate one robot from the rest when they are all stuck at the same spot.
Is the problem undecidable for gadgets without this ability? What about classes of gadgets that have already been studied such as self-closing doors or reversible, deterministic gadgets?

In the 0-player model with spawners we investigated a synchronous model for the robots where they all took turns making their moves. One could imagine asking about various asynchronous models of robot motion through the gadgets.

For 1-player multi-agent motion planning, we investigated robot reachability and multi-agent targeted reconfiguration. The hardness for both these problems relies on simulating Petri nets with a symmetric self-closing door.
Do there exist reversible gadgets for which the problem is the same complexity?
How does this relate to reversible Petri nets?

We also did not investigate spawners in the 2-player setting. It seems likely that this problems is Undecideable for many gadget; however, the 0-player and 1-player constructions do not obviously adapt to give this result. 

Finally, in the 2-player impartial case, does the complexity change for other gadgets?
Are there any gadgets for which finding a winning strategy is provably easier? What about cases where the impartial game is harder than the regular 2-player game?

\bibliographystyle{alpha}
\bibliography{refs}

\newcommand{\etalchar}[1]{$^{#1}$}
\begin{thebibliography}{BMLC{\etalchar{+}}19}

\bibitem[ABD{\etalchar{+}}20]{doors}
Joshua Ani, Jeffrey Bosboom, Erik~D. Demaine, Yevhenii Diomidov, Dylan Hendrickson, and Jayson Lynch.
\newblock Walking through doors is hard, even without staircases: Proving {PSPACE}-hardness via planar assemblies of door gadgets.
\newblock In {\em Proceedings of the 10th International Conference on Fun with Algorithms (FUN 2020)}, Favignana, Italy, September 2020.

\bibitem[ACD{\etalchar{+}}22]{ani2022checked}
Joshua Ani, Lily Chung, Erik~D. Demaine, Yevhenii Diomidov, Dylan Hendrickson, and Jayson Lynch.
\newblock Pushing blocks via checkable gadgets: {PSPACE}-completeness of {Push-1F} and {Block/Box Dude}.
\newblock In {\em Proceedings of the 11th International Conference on Fun with Algorithms}, pages 2:1--2:30, Island of Favignana, Sicily, Italy, May--June 2022.

\bibitem[ADD{\etalchar{+}}22]{ani2022traversability}
Joshua Ani, Erik~D. Demaine, Yevhenii Diomidov, Dylan~H. Hendrickson, and Jayson Lynch.
\newblock Traversability, reconfiguration, and reachability in the gadget framework.
\newblock In Petra Mutzel, Md.~Saidur Rahman, and Slamin, editors, {\em Proceedings of the 16th International Conference and Workshops on Algorithms and Computation}, volume 13174 of {\em Lecture Notes in Computer Science}, pages 47--58, Jember, Indonesia, March 2022.

\bibitem[ADG{\etalchar{+}}21]{a2021characterizing}
Hugo~A. Akitaya, Erik~D. Demaine, Andrei Gonczi, Dyl an~H.~Hendrickson, Adam Hesterberg, Matias Korman, Oliver Korten, Ja~yson Lynch, Irene Parada, and Vera Sacrist{\'a}n.
\newblock Characterizing universal reconfigurability of modular pivoting robots.
\newblock In Kevin Buchin and \'Eric~Colin de~Verdi\`ere, editors, {\em Proceedings of the 37th International Symposium on Computational Geometry}, LIPIcs, pages 10:1--10:20, 2021.

\bibitem[ADHL22]{ani2022trains}
Joshua Ani, Erik~D. Demaine, Dylan Hendrickson, and Jayson Lynch.
\newblock Trains, games, and complexity: 0/1/2-player motion planning through input/output gadgets.
\newblock In Petra Mutzel, Md.~Saidur Rahman, and Slamin, editors, {\em Proceedings of the 16th International Conference and Workshops on Algorithms and Computation}, volume 13174 of {\em Lecture Notes in Computer Science}, pages 187--198, Jember, Indonesia, March 2022.

\bibitem[AFG{\etalchar{+}}22]{alaniz2022reachability}
Robert~M. Alaniz, Bin Fu, Timothy Gomez, Elise Grizzell, Andrew Rodriguez, Robert Schweller, and Tim Wylie.
\newblock Reachability in restricted chemical reaction networks.
\newblock {\em arXiv preprint arXiv:2211.12603}, 2022.

\bibitem[BMLC{\etalchar{+}}19]{balanza2019full}
Jose Balanza-Martinez, Austin Luchsinger, David Caballero, Rene Reyes, Angel~A Cantu, Robert Schweller, Luis~Angel Garcia, and Tim Wylie.
\newblock Full tilt: Universal constructors for general shapes with uniform external forces.
\newblock In {\em Proceedings of the 30th Annual ACM-SIAM Symposium on Discrete Algorithms}, pages 2689--2708. SIAM, 2019.

\bibitem[CCG{\etalchar{+}}20]{caballero2020relocating}
David Caballero, Angel~A. Cantu, Timothy Gomez, Austin Luchsinger, Robert Schweller, and Tim Wylie.
\newblock Relocating units in robot swarms with uniform control signals is {PSPACE}-complete.
\newblock {\em CCCG 2020}, 2020.

\bibitem[CO22]{czerwinski2022reachability}
Wojciech Czerwi{\'n}ski and {\L}ukasz Orlikowski.
\newblock Reachability in vector addition systems is {A}ckermann-complete.
\newblock In {\em Proceedings of the 62nd Annual IEEE Symposium on Foundations of Computer Science}, pages 1229--1240, 2022.

\bibitem[DGLR18]{gadgets}
Erik~D. Demaine, Isaac Grosof, Jayson Lynch, and Mikhail Rudoy.
\newblock Computational complexity of motion planning of a robot through simple gadgets.
\newblock In {\em Proceedings of the 9th International Conference on Fun with Algorithms}, pages 18:1--18:21, La Maddalena, Italy, June 2018.

\bibitem[DHHL22]{demaine2022pspace}
Erik~D. Demaine, Robert~A. Hearn, Dylan Hendrickson, and Jayson Lynch.
\newblock {PSPACE}-completeness of reversible deterministic systems.
\newblock In {\em Proceedings of the 9th Conference on Machines, Computations and Universality}, pages 91--108, Debrecen, Hungary, August--September 2022.

\bibitem[DHL20]{gadgets2}
Erik~D. Demaine, Dylan Hendrickson, and Jayson Lynch.
\newblock Toward a general theory of motion planning complexity: Characterizing which gadgets make games hard.
\newblock In {\em Proceedings of the 11th Conference on Innovations in Theoretical Computer Science}, Seattle, Washington, January 2020.

\bibitem[Esp05]{esparza2005decidability}
Javier Esparza.
\newblock Decidability and complexity of petri net problems—an introduction.
\newblock {\em Lectures on Petri Nets I: Basic Models: Advances in Petri Nets}, pages 374--428, 2005.

\bibitem[Hen21]{hendrickson2021gadgets}
Dylan Hendrickson.
\newblock Gadgets and gizmos: A formal model of simulation in the gadget framework for motion planning.
\newblock Master's thesis, Massachusetts Institute of Technology, 2021.

\bibitem[Ler22]{leroux2022reachability}
J{\'e}r{\^o}me Leroux.
\newblock The reachability problem for {P}etri nets is not primitive recursive.
\newblock In {\em Proceedings of the 62nd Annual IEEE Symposium on Foundations of Computer Science}, pages 1241--1252, 2022.

\bibitem[LS19]{leroux2019reachability}
J{\'e}r{\^o}me Leroux and Sylvain Schmitz.
\newblock Reachability in vector addition systems is primitive-recursive in fixed dimension.
\newblock In {\em Proceedings of the 34th Annual ACM/IEEE Symposium on Logic in Computer Science}, pages 1--13. IEEE, 2019.

\bibitem[Lyn20]{lynch2020framework}
Jayson Lynch.
\newblock {\em A framework for proving the computational intractability of motion planning problems}.
\newblock PhD thesis, Massachusetts Institute of Technology, 2020.

\bibitem[Min67]{counter}
Marvin Minsky.
\newblock {\em Computation: Finite and Infinite Machines}.
\newblock Prentice-Hall, Inc, 1967.

\bibitem[Rac78]{rackoff1978covering}
Charles Rackoff.
\newblock The covering and boundedness problems for vector addition systems.
\newblock {\em Theoretical Computer Science}, 6(2):223--231, 1978.

\bibitem[SC79]{gx}
Larry~J. Stockmeyer and Ashok~K. Chandra.
\newblock Provably difficult combinatorial games.
\newblock {\em Siam Journal on Computing}, 8(2):151--174, May 1979.

\end{thebibliography}
\end{document}